\documentclass[11pt]{article}

\usepackage{amsmath}
\usepackage{amssymb}
\usepackage{amsthm}
\usepackage{graphics}
\usepackage{psfrag}
\usepackage{times}

\newtheorem{theorem}{Theorem}
\newtheorem{lemma}{Lemma}
\newtheorem{corollary}{Corollary}
\newtheorem{proposition}{Proposition}
\newtheorem{conjecture}{Conjecture}

\newcommand{\old}[1]{{}}
\newcommand{\etal}{et~al.}

\newcommand{\RR}{\mathbb{R}}
\newcommand{\ZZ}{\mathbb{Z}}

\def\F{\mathcal F}
\def\I{\mathcal I}
\def\L{\mathcal L}

\def\S{\mathcal S}
\def\T{\mathcal T}

\begin{document}

\title{\bf Piercing translates and homothets of a convex body\footnote{%
A preliminary version of this paper appeared in the
Proceedings of the 17th Annual European Symposium on Algorithms (ESA 2009),
pages 131--142.}}

\author{%
Adrian Dumitrescu\footnote{Department of Computer Science,
University of Wisconsin--Milwaukee,
WI 53201-0784, USA\@.
Email: \texttt{ad@cs.uwm.edu}.
Supported in part by NSF CAREER grant CCF-0444188.}
\and
Minghui Jiang\footnote{Department of Computer Science,
Utah State University, Logan, UT 84322-4205, USA\@.
Email: \texttt{mjiang@cc.usu.edu}.
Supported in part by NSF grant DBI-0743670.}}

\maketitle

\begin{abstract}
According to a classical result of Gr\"unbaum,
the transversal number $\tau(\F)$
of any family $\F$ of pairwise-intersecting translates or homothets
of a convex body $C$ in $\RR^d$
is bounded by a function of $d$.
Denote by $\alpha(C)$ (resp.~$\beta(C)$) the supremum of the ratio
of the transversal number $\tau(\F)$ to the packing number $\nu(\F)$
over all families $\F$ of translates (resp.~homothets)
of a convex body $C$ in $\RR^d$.
Kim \etal\ recently showed that
$\alpha(C)$ is bounded by a function of $d$
for any convex body $C$ in $\RR^d$,
and gave the first bounds on $\alpha(C)$ for convex bodies $C$ in $\RR^d$
and on $\beta(C)$ for convex bodies $C$ in the plane.

Here
we show that $\beta(C)$ is also bounded by a function of $d$
for any convex body $C$ in $\RR^d$,
and present new or improved bounds on both $\alpha(C)$ and $\beta(C)$
for various convex bodies $C$ in $\RR^d$ for all dimensions $d$. 
Our techniques explore interesting inequalities
linking the covering and packing densities of a convex body. 
Our methods for obtaining upper bounds are constructive and lead to
efficient constant-factor 
approximation algorithms for finding a minimum-cardinality point set that
pierces a set of translates or homothets of a convex body. 
\end{abstract}



\section{Introduction}

A \emph{convex body} is a compact convex set in $\RR^d$ with nonempty interior.
Let $\F$ be a family of convex bodies.
The \emph{packing number} $\nu(\F)$
is the maximum cardinality of a set of pairwise-disjoint convex bodies in $\F$,
and the \emph{transversal number} $\tau(\F)$
is the minimum cardinality of a set of points that intersects
every convex body in $\F$.

Let $G$ be the \emph{intersection graph} of $\F$
with one vertex for each convex body in $\F$
and with an edge between two vertices
if and only if the two corresponding convex bodies intersect.
The \emph{independence number} $\alpha(G)$
is the maximum cardinality of an independent set in $G$.
The \emph{clique partition number} $\vartheta(G)$
is the minimum number of classes
in a partition of the vertices of $G$ into cliques.
Since a set of pairwise-disjoint convex bodies in $\F$ corresponds to
an independent set in $G$, we have $\nu(\F) = \alpha(G)$.
Also, since any subset of convex bodies in $\F$ that share a common point
corresponds to a clique in $G$, we have $\tau(\F) \ge \vartheta(G)$.
For the special case that $\F$ is a family of axis-parallel boxes in $\RR^d$,
we indeed have $\tau(\F) = \vartheta(G)$
since any subset of pairwise-intersecting boxes share a common point.
In general, we clearly have the inequality $\vartheta(G) \ge \alpha(G)$,
thus also $\tau(\F) \ge \nu(\F)$. But what else can be said about the
relation between $\tau(\F)$ and $\nu(\F)$?

\begin{figure}[htbp]
\centering
\includegraphics{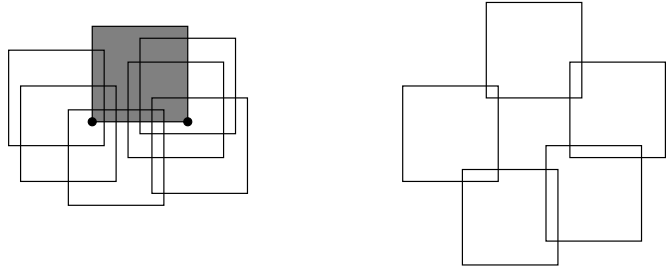}
\caption{\small
Piercing a family $\F$ of axis-parallel unit squares.
Left: all squares that intersect the highest (shaded) square
contain one of its two lower vertices.
Right: five squares form a 5-cycle.}
\label{fig:s}
\end{figure}

For example, let $\F$ be any family of axis-parallel unit squares in the plane,
and refer to Figure~\ref{fig:s}.
One can obtain a subset of pairwise-disjoint squares
by repeatedly selecting the highest square that does not intersect the
previously selected squares.
Then $\F$ is pierced by the set of points
consisting of the two lower vertices of each square in the subset.
This implies that $\tau(\F) \le 2\cdot \nu(\F)$.
The factor of $2$ cannot be improved below $\frac32$ since $\tau(\F) =
3$ and $\nu(\F) = 2$ for a family $\F$ of five squares arranged into a
5-cycle~\cite{GL85}.

For a convex body $C$ in $\RR^d$, $d \ge 2$,
define
$$
\alpha(C) = \sup_{\F_t} \frac{\tau(\F_t)}{\nu(\F_t)}
\quad\textrm{and}\quad
\beta(C) = \sup_{\F_h} \frac{\tau(\F_h)}{\nu(\F_h)},
$$
where $\F_t$ ranges over all families of translates of $C$,
and $\F_h$ ranges over all families of (positive) homothets of $C$.
Our previous discussion (Figure~\ref{fig:s}) yields the bounds
$\frac32 \le \alpha(C) \le 2$ for any square $C$.

Define $\alpha_1(C)$ (resp.~$\beta_1(C)$) as the smallest number $k$ such that
for any family $\F$ of \emph{pairwise-intersecting}
translates (resp.~homothets) of a convex body $C$,
there exists a set of $k$ points that intersects every member of $\F$.
Note that $\alpha$ and $\beta$ generalize $\alpha_1$ and $\beta_1$.
For any convex body $C$,
the four numbers $\alpha(C)$, $\beta(C)$, $\alpha_1(C)$, and $\beta_1(C)$
are invariant under any non-singular affine transformation of $C$,
and we have the four inequalities 
$\alpha_1(C) \le \alpha(C)$, 
$\beta_1(C) \le \beta(C)$,
$\alpha_1(C) \le \beta_1(C)$,
and $\alpha(C) \le \beta(C)$.

Gr\"unbaum~\cite{Gr59} showed that,
for any convex body $C$ in $\RR^d$,
both $\alpha_1(C)$ and $\beta_1(C)$ are bounded by functions of $d$.
Deriving bounds on $\alpha_1(C)$ and $\beta_1(C)$ for various types
of convex bodies $C$ in $\RR^d$ is typical of classic Gallai-type
problems~\cite{DGK63,We04}, and has been extensively studied.
For example, a result by Karasev~\cite{Ka00} states that
$\alpha_1(C) \le 3$ for any convex body $C$ in the plane,
i.e., for any family of pairwise-intersecting translates
of a convex body in the plane, there always exists a set of three
points that intersects every member of the family.
It is folklore that
$\alpha_1(C) = \beta_1(C) = 1$ for any parallelogram $C$
	(see~\cite{Gr59} and the references therein).
Also,
$\alpha_1(C) = 2$ for any affinely regular hexagon $C$~\cite[Example~2]{Gr59},
$\alpha_1(C) = \beta_1(C) = 3$ for any triangle $C$~\cite{CS67}\footnote{%
	We give a simpler construction for the lower bound
	$\alpha_1(C) \ge 3$ for any triangle $C$ in Appendix~\ref{sec:t1}.},
$\alpha_1(C) = 3 < 4 = \beta_1(C)$ for any (circular) disk $C$~\cite{Gr59,Da86},
and
$\beta_1(C) \le 7$ for any centrally symmetric convex body $C$
in the plane~\cite{Gr59}.
Perhaps the most celebrated result on point transversals of convex sets
is Alon and Kleitman's solution to the Hadwiger-Debrunner
$(p,q)$-problem~\cite{AK92}.
We refer to the two surveys~\cite[pp.~142--150]{DGK63}
and~\cite[pp.~77--78]{We04}
for more related results.

The two numbers $\alpha_1(C)$ and $\beta_1(C)$ bound the values of $\tau(\F)$
for special families $\F$ of translates and homothets, respectively,
of a convex body $C$ with $\nu(\F) = 1$.
It is thus natural to study the general case $\nu(\F) \ge 1$,
and to obtain estimates on $\alpha(C)$ and $\beta(C)$.
Despite the many previous bounds on $\alpha_1(C)$
and $\beta_1(C)$~\cite{DGK63,We04}, 
first estimates on $\alpha(C)$ and $\beta(C)$ have been only obtained
recently~\cite{KNPS06}.
Note that the related problem for families of $d$-intervals (which are nonconvex)
has been extensively studied~\cite{Ta95,KT96,Ka97,Al98,Ma01}.

Kim \etal~\cite{KNPS06} showed that
$\alpha(C)$ is bounded by a function of $d$
for any convex body $C$ in $\RR^d$,
and gave the first bounds on $\alpha(C)$ for convex bodies $C$ in $\RR^d$
and on $\beta(C)$ for convex bodies $C$ in the plane.
In this paper,
we show that $\beta(C)$ is also bounded by a function of $d$
for any convex body $C$ in $\RR^d$,
and present new or improved bounds on both $\alpha(C)$ and $\beta(C)$
for various convex bodies $C$ in $\RR^d$ for all dimensions $d$. 

Note that in the definitions of $\alpha$ and $\beta$,
both the convexity of $C$ and the homothety of $\F_t$ and $\F_h$
are necessary for the values $\alpha(C)$ and $\beta(C)$ to be bounded.
To see the necessity of convexity, let $C$ be the union of
a vertical line segment with endpoints $(0,0)$ and $(0,1)$
and a horizontal line segment with endpoints $(0,0)$ and $(1,0)$,
where the shared endpoint $(0,0)$ is the \emph{corner},
and let $\F$ be a family of $n$ translates of $C$
with corners at $(i/n, -i/n)$, $1 \le i \le n$~\cite{GL85}.
Then at least $\lceil n/2 \rceil$ points are required
to intersect every member of $\F$. To see the necessity of homothety,
let $\F$ be a family of $n$ pairwise intersecting line segments
(or very thin rectangles) in the plane such that no three have a common point.
Then again at least $\lceil n/2 \rceil$ points are required to
intersect every member of $\F$. 

\paragraph{Definitions.}

For a convex body $C$ in $\RR^d$,
denote by $|C|$ the Lebesgue measure of $C$,
i.e., the area in the plane, or the volume in $d$-space for $d \ge 3$.
For a family $\F$ of convex bodies in $\RR^d$,
denote by $|\F|$ the Lebesgue measure of the union of the convex bodies in
$\F$, i.e., $|\bigcup_{C \in \F} C|$.

For two convex bodies $A$ and $B$ in $\RR^d$,
denote by $A+B = \{ a+b \mid a\in A, b\in B \}$ the Minkowski sum
of $A$ and $B$. For a convex body $C$ in $\RR^d$,
denote by $\lambda C = \{ \lambda c \mid c \in C \}$
the \emph{scaled copy} of $C$ by a factor of $\lambda \in \RR$,
denote by $-C = \{ -c \mid c \in C \}$
the \emph{reflexion} of $C$ about the origin,
and denote by $C+a = \{ c+a \mid c\in C \}$
the \emph{translate} of $C$ by the vector from the origin to $a$.
Write $C-C$ for $C + (-C)$.

For two parallelepipeds $P$ and $Q$ in $\RR^d$ that are parallel to each other
(but are not necessarily axis-parallel or orthogonal),
denote by $\lambda_i(P,Q)$, $1 \le i \le d$,
the length ratios of the edges of $Q$ to
the corresponding parallel edges of $P$.
Then, for a convex body $C$ in $\RR^d$, define
$$
\gamma(C) = \min_{P,Q}
\left(
\left\lceil \lambda_d(P,Q) \right\rceil
\prod_{i=1}^{d-1}
\left\lceil \lambda_i(P,Q) + 1 \right\rceil
\right),
$$
where $P$ and $Q$ range over all pairs of parallelepipeds in $\RR^d$
that are parallel to each other, such that $P \subseteq C \subseteq Q$.
Note that in this case $\lambda_i(P,Q) \geq 1$ for $1 \le i \le d$.

We review some standard definitions of packing and covering densities
in the following; see~\cite[Chapter~1]{BMP05}.
A family $\F$ of convex bodies is a \emph{packing}
in a domain $Y \subseteq \RR^d$ if
$\bigcup_{C \in \F} C \subseteq Y$
and the convex bodies in $\F$ are pairwise-disjoint;
$\F$ is a \emph{covering} of $Y$ if $Y \subseteq \bigcup_{C \in \F} C$.
The \emph{density} of a family $\F$ relative to a bounded domain $Y$ is
$\rho(\F, Y) = (\sum_{C \in \F} |C|)/|Y|$.
If $Y = \RR^d$ is the whole space, then the \emph{upper density}
and the \emph{lower density} of $\F$ are, respectively,
$$
\overline\rho(\F, \RR^d) = \limsup_{r\to\infty} \rho(\F, B^d(r))
\quad\textup{and}\quad
\underline\rho(\F, \RR^d) = \liminf_{r\to\infty} \rho(\F, B^d(r)),
$$
where $B^d(r)$ denote a ball of radius $r$ centered at the origin
(since we are taking the limit as $r\to\infty$,
a hypercube of side length $r$ can be used instead of a ball of radius $r$).
For a convex body $C$ in $\RR^d$,
define the \emph{packing density} of $C$ as
$$
\delta(C) = \sup_{\F~\textrm{packing}} \overline\rho(\F, \RR^d),
$$
where $\F$ ranges over all packings in $\RR^d$ with congruent copies of $C$,
and define the \emph{covering density} of $C$ as
$$
\theta(C) = \inf_{\F~\textrm{covering}} \underline\rho(\F, \RR^d),
$$
where $\F$ ranges over all coverings of $\RR^d$ with congruent copies of $C$.
If the members of $\F$ are restricted to translates of $C$,
then we have the \emph{translative} packing and covering densities
$\delta_T(C)$ and $\theta_T(C)$.
If the members of $\F$ are further restricted to translates of $C$
by vectors of a lattice,
then we have the \emph{lattice} packing and covering densities
$\delta_L(C)$ and $\theta_L(C)$.
Note that the four densities
$\theta_T(C)$, $\theta_L(C)$, $\delta_T(C)$, and $\delta_L(C)$
are invariant under any non-singular affine transformation of $C$.
For any convex body $C$ in $\RR^d$,
we have the inequalities
$\delta_L(C) \le \delta_T(C) \le \delta(C) \le 1 \le
\theta(C) \le \theta_T(C) \le \theta_L(C)$.

For two convex bodies $A$ and $B$ in $\RR^d$,
denote by $\kappa(A,B)$ the smallest number $\kappa$
such that $A$ can be covered by $\kappa$ translates of $B$.

\paragraph{Main results.}

Kim \etal~\cite{KNPS06} recently proved that,
for any family $\F$ of translates of a convex body in $\RR^d$,
$\tau(\F) \le 2^{d-1}d^d \cdot \nu(\F)$,
in particular $\tau(\F) \le 108 \cdot \nu(\F)$ when $d = 3$,
and moreover $\tau(\F) \le 8 \cdot \nu(\F) - 5$ when $d = 2$.
We improve these bounds for all dimensions $d$ in the following theorem:

\begin{theorem}\label{thm:conv}
For any family $\F$ of translates of a convex body $C$ in $\RR^d$,
\begin{equation}\label{eq:newconvd}
\tau(\F) \le \gamma(C) \cdot \nu(\F),
	\quad\textrm{ where } \gamma(C) \le d(d+1)^{d-1}.
\end{equation}
In particular,
$\tau(\F) \le 48 \cdot \nu(\F)$ when $d = 3$,
and $\tau(\F) \le 6 \cdot \nu(\F)$ when $d = 2$.
\end{theorem}

For any parallelepiped $C$ in $\RR^d$,
we can choose two parallelepipeds $P$ and $Q$ such that $P = Q = C$
hence $P \subseteq C \subseteq Q$.
Then $\lambda_i(P,Q) = 1$ for $1 \le i \le d$,
and $\gamma(C) = 2^{d-1}$.
This implies the following corollary:

\begin{corollary}\label{cor:parallelepipedt}
For any family $\F$ of translates of a parallelepiped in $\RR^d$,
$\tau(\F) \le 2^{d-1} \cdot \nu(\F)$.
\end{corollary}

In contrast, for a family $\F$ of (not necessarily congruent or similar)
axis-parallel parallelepipeds in $\RR^d$,
the current best upper bound~\cite{FK93} (see also~\cite{Ka91,KT96,Ni01}) is
$$
\tau(\F) \le \nu(\F)\log^{d-2}\nu(\F)(\log\nu(\F) - 1/2) + d.
$$

Kim \etal~\cite{KNPS06} also proved that,
for any family $\F$ of translates
of a centrally symmetric convex body in the plane,
$\tau(\F) \le 6 \cdot \nu(\F) - 3$.
The following theorem gives
a general bound for any centrally symmetric convex body in $\RR^d$
and an improved bound (if $\nu(\F) \ge 5$)
for any centrally symmetric convex body in the plane:

\begin{theorem}\label{thm:symm}
For any family $\F$ of translates of a centrally symmetric convex body $S$
in $\RR^d$,
\begin{equation}\label{eq:newsymmd}
\tau(\F) \le 2^d \cdot \frac{\theta_L(S)}{\delta_L(S)} \cdot \nu(\F).
\end{equation}
Moreover,
$\tau(\F) \le 24 \cdot \nu(\F)$ when $d = 3$,
and
$\tau(\F) \le \frac{16}3 \cdot \nu(\F)$ when $d = 2$.
\end{theorem}

For special types of convex bodies in the plane,
the following theorem gives sharper bounds
than the bounds implied by Theorem~\ref{thm:conv} and Theorem~\ref{thm:symm}.
Also, as we will show later,
inequality \eqref{eq:greedyt} below may give a better asymptotic bound
than \eqref{eq:newconvd} and \eqref{eq:newsymmd}
for high dimensions.

\begin{theorem}\label{thm:greedyt}
Let $\F$ be a family of translates of a convex body $C$ in $\RR^d$.
Then
\begin{equation}\label{eq:greedyt}
\tau(\F) \le \min_L \kappa((C-C) \cap L, C) \cdot \nu(\F),
\end{equation}
where $L$ ranges over all closed half spaces bounded by hyperplanes
through the center of $C-C$.
Moreover, $\tau(\F) \le 4 \cdot \nu(\F) - 1$
if $C$ is a centrally symmetric convex body in the plane.
Also,
\begin{enumerate}
\item[\textup{(i)}]
If $C$ is a square, then
$\tau(\F) \le 2 \cdot \nu(\F) - 1$,
\item[\textup{(ii)}]
If $C$ is a triangle, then
$\tau(\F) \le 5 \cdot \nu(\F) - 2$,
\item[\textup{(iii)}]
If $C$ is a disk, then
$\tau(\F) \le 4 \cdot \nu(\F) - 1$.
\end{enumerate}
\end{theorem}

Having presented our bounds for families of translates, we now turn to
families of homothets. Kim \etal~\cite{KNPS06} proved that,
for any family $\F$ of homothets of a convex body $C$ in the plane,
$\tau(\F) \le 16 \cdot \nu(\F)$ and, if $C$ is centrally symmetric,
$\tau(\F) \le 9 \cdot \nu(\F)$.
The following theorem
gives a general bound for any convex body in $\RR^d$,
an improved bound for any centrally symmetric convex body in the plane,
and additional bounds for special types of convex bodies in the plane:

\begin{theorem}\label{thm:greedyh}
Let $\F$ be a family of homothets of a convex body $C$ in $\RR^d$. Then
\begin{equation}\label{eq:greedyh}
\tau(\F) \le \kappa(C-C, C) \cdot \nu(\F).
\end{equation}
In particular, $\tau(\F) \le 7 \cdot \nu(\F)$
if $C$ is a centrally symmetric convex body in the plane.
Moreover,
\begin{enumerate}
\item[\textup{(i)}]
If $C$ is a square, then
$\tau(\F) \le 4 \cdot \nu(\F) - 3$,
\item[\textup{(ii)}]
If $C$ is a triangle, then
$\tau(\F) \le 12 \cdot \nu(\F) - 9$,
\item[\textup{(iii)}]
If $C$ is a disk, then
$\tau(\F) \le 7 \cdot \nu(\F) - 3$.
\end{enumerate}
\end{theorem}

For any parallelepiped $C$ in $\RR^d$,
$C-C$ is a translate of $2C$ and can be covered by $2^d$ translates of $C$,
thus $\kappa(C-C,C) \le 2^d$.
This implies the following corollary:

\begin{corollary}\label{cor:parallelepipedh}
For any family $\F$ of homothets of a parallelepiped in $\RR^d$,
$\tau(\F) \le 2^d \cdot \nu(\F)$.
\end{corollary}

Both Theorem~\ref{thm:greedyt} and Theorem~\ref{thm:greedyh}
are obtained by a simple greedy method,
used also previously by Kim \etal~\cite{KNPS06}.
Although we have improved their bounds using new techniques
in Theorem~\ref{thm:conv} and Theorem~\ref{thm:symm},
we show that a refined analysis of the simple greedy method
yields even better asymptotic bounds for high dimensions
in Theorem~\ref{thm:greedyt} and Theorem~\ref{thm:greedyh}.
We will use the following lemma by Chakerian and Stein~\cite{CS67}
in our analysis:

\begin{lemma}[Chakerian and Stein~\cite{CS67}]\label{lem:cs}
For every convex body $C$ in $\RR^d$ there exist two parallelepipeds
$P$ and $Q$ such that $P \subseteq C \subseteq Q$,
where $P$ and $Q$ are homothetic with ratio at most $d$.
\end{lemma}

For any convex body $C$ in $\RR^d$,
let $P$ and $Q$ be the two parallelepipeds in Lemma~\ref{lem:cs}.
Since
$C-C \subseteq Q-Q$ and $P \subseteq C$,
it follows that
$\kappa(C-C,C) \le \kappa(Q-Q,P) = \kappa(2Q,P) \le (2d)^d$;
see also~\cite[Lemma~4]{KNPS06}.
The classic survey by Danzer, Gr\"unbaum, and Klee~\cite[pp.~146--147]{DGK63}
lists several other upper bounds due to Rogers 
and Danzer:
(i)
$\kappa(C-C,C) \le \frac{2^d}{d+1} 3^{d+1} \theta_T(C)$
for any convex body $C$ in $\RR^d$,
(ii)
$\kappa(C-C,C) \le 5^d$ and $\kappa(C-C,C) \le 3^d \theta_T(C)$
for any centrally symmetric convex body $C$ in $\RR^d$.
Note that $\theta_T(C) < d\ln d + d\ln\ln d + 5d = O(d\log d)$
for any convex body $C$ in $\RR^d$, according to a result of Rogers~\cite{Ro57}.
The following lemma summarizes the upper bounds on $\kappa(C-C,C)$:

\old{
We now prove the bound $\kappa(C-C,C) \le 5^d$
for any centrally symmetric convex body $C$ in $\RR^d$
Without loss of generality, assume that the center of $C$ is at the origin.
Since $C$ is centrally symmetric, we have $C-C = 2C$.
Now obtain an arbitrary maximal packing in $2.5C$ by translates of $0.5C$,
then expand each translate of $0.5C$ in the packing by a factor of $2$.
We claim that the resulting translates of $C$ must cover $2C$.
Suppose for contradiction that a point $p \in 2C$ is not covered.
Consider the translate $C + q$ corresponding to
any translate $0.5C + q$ in the packing centered at a point $q$.
Since $p$ is not contained in
$C + q = (0.5C + 0.5C) + q = (0.5C - 0.5C) + q$,
it follows by Lemma~\ref{lem:ccc} (ii)
(using the center of $C$ as the reference point)
that $0.5C + p$ does not intersect $0.5C + q$.
Furthermore, $0.5C + p$ is contained in $0.5C + 2C = 2.5C$.
This implies that we can add $0.5C + p$ to the packing,
which contradicts its maximality.
Therefore the minimum number of translates of $C$ that cover $2C$
is at most $(\frac{2.5}{0.5})^d = 5^d$ by a volume argument,
and we have $\kappa(C-C,C) = \kappa(2C,C) \le 5^d$.
}

\begin{lemma}\label{lem:kappa}
For any convex body $C$ in $\RR^d$,
$\kappa(C-C,C) \le \min\{
		(2d)^d,\,
		\frac{2^d}{d+1} 3^{d+1} \theta_T(C)
	\} = O(6^d \log d)$.
Moreover,
if $C$ is centrally symmetric, then
$\kappa(C-C,C) \le \min\{
		5^d,\,
		3^d \theta_T(C)
	\} = O(3^d d\log d)$.
\end{lemma}

From Lemma~\ref{lem:kappa} and Theorem~\ref{thm:greedyh},
it follows that $\beta(C)$ is bounded by a function of $d$, namely by
$O(6^d \log d)$, for any convex body $C$ in $\RR^d$.
Since $\min_L \kappa((C-C) \cap L, C) \le \kappa(C-C,C)$,
Lemma~\ref{lem:kappa} also provides upper bounds on
$\min_L \kappa((C-C) \cap L, C)$ in Theorem~\ref{thm:greedyt}.
As a result,
\eqref{eq:greedyt} implies
an upper bound $\tau(\F) \le O(6^d\log d) \cdot \nu(\F)$
for any family $\F$ of translates of a convex body in $\RR^d$,
which is better than the upper bound
$\tau(\F) \le d(d+1)^{d-1} \cdot \nu(\F)$ in \eqref{eq:newconvd}
when $d$ is sufficiently large.
Also, \eqref{eq:greedyt} implies
an upper bound $\tau(\F) \le 3^d\theta_T(S) \cdot \nu(\F)$
for any family $\F$ of translates of a centrally symmetric convex body $S$
in $\RR^d$.
Schmidt~\cite{Sc63} showed that,
for any centrally symmetric convex body $S$,
$\delta_L(S) = \Omega(d/2^d)$; see also~\cite[p.~12]{BMP05}.
Hence \eqref{eq:newsymmd} implies the bound
$\tau(\F) \le O(4^d/d) \theta_L(S) \cdot \nu(\F)$.
Note that $\theta_T(S) \le \theta_L(S)$.
So \eqref{eq:greedyt} may be also better than \eqref{eq:newsymmd}
for high dimensions.
Table~\ref{tab1} summarizes the current best upper bounds
on $\alpha(C)$ and $\beta(C)$ (obtained by us and by others)
for various types of convex bodies $C$ in $\RR^d$.

\begin{table}[htbp]
\centering
\begin{tabular}{lr|lr}
\hline
\multicolumn{2}{c|}{\textsl{Convex body $C$ in $\RR^d$}}
	& \multicolumn{2}{c}{\textsl{$\alpha(C)$ upper}}
\\
\hline
arbitrary & $d=2$
	& $6$ & T\ref{thm:conv}
\\
centr.~symm. & $d=2$
	& $4$ & T\ref{thm:greedyt}
\\
arbitrary & $d=3$
	& $48$ & T\ref{thm:conv}
\\
centr.~symm. & $d=3$
	& $24$ & T\ref{thm:symm}
\\
arbitrary & $d>3$
	& $\min\{
		d(d+1)^{d-1},\,
		\frac{2^d}{d+1} 3^{d+1} \theta_T(C) \}$
	& T\ref{thm:conv} T\ref{thm:greedyh}-L\ref{lem:kappa}
\\
centr.~symm. & $d>3$
	& $\min\{
		d(d+1)^{d-1},\,
		2^d \frac{\theta_L(C)}{\delta_L(C)},\,
		5^d,\,
		3^d \theta_T(C) \}$
	& T\ref{thm:conv} T\ref{thm:symm} T\ref{thm:greedyh}-L\ref{lem:kappa}
\\
parallelepiped & $d \ge 2$
	& $2^{d-1}$ & C\ref{cor:parallelepipedt}
\\
\hline
\multicolumn{2}{c|}{\textsl{Convex body $C$ in $\RR^d$}}
	& \multicolumn{2}{c}{\textsl{$\beta(C)$ upper}}
\\
\hline
arbitrary & $d=2$
	& $16$ & \cite{KNPS06}
\\
centr.~symm. & $d=2$
	& $7$ & T\ref{thm:greedyh}
\\
arbitrary & $d=3$
	& $216$ & $^\dag$T\ref{thm:greedyh}-L\ref{lem:kappa}
\\
centr.~symm. & $d=3$
	& $125$ & $^\dag$T\ref{thm:greedyh}-L\ref{lem:kappa}
\\
arbitrary & $d>3$
	& $\min\{
		(2d)^d,\,
		\frac{2^d}{d+1} 3^{d+1} \theta_T(C) \}$
	& T\ref{thm:greedyh}-L\ref{lem:kappa}
\\
centr.~symm. & $d>3$
	& $\min\{
		5^d,\,
		3^d \theta_T(C) \}$
	& T\ref{thm:greedyh}-L\ref{lem:kappa}
\\
parallelepiped & $d \ge 2$
	& $2^d$ & C\ref{cor:parallelepipedh}
\\
\hline
\end{tabular}
\caption{\small Upper bounds on $\alpha(C)$ and $\beta(C)$
for a convex body $C$ in $\RR^d$.
$^\dag$By Theorem~\ref{thm:greedyh} and Lemma~\ref{lem:kappa}:
for $d = 3$, $(2d)^d = 216$ and $5^d = 125$.}
\label{tab1}
\end{table}

A natural question is whether $\alpha(C)$ or $\beta(C)$ need to be
exponential in $d$. The following theorem gives a positive answer:

\begin{theorem}\label{thm:lower}
For any convex body $C$ in $\RR^d$,
$\beta(C) \ge \alpha(C) \ge \frac{\theta_T(C)}{\delta_T(C)}$.
In particular,
if $C$ is the unit ball $B^d$ in $\RR^d$,
then $\beta(C) \ge \alpha(C) \ge 2^{(0.599 \pm o(1))d}$ 
as $d\to\infty$.
\end{theorem}

Kim \etal~\cite{KNPS06} asked whether the upper bound
$\tau(\F) \le 3 \cdot \nu(\F)$
holds for any family $\F$ of translates of a centrally symmetric convex body
in the plane.
This upper bound, if true, is best possible because
there exists a family $\F$ of congruent disks
(i.e., translates of a disk)
such that $\tau(\F) = 3 \cdot \nu(\F)$ for any $\nu(\F) \ge 1$~\cite{Gr59};
see also~\cite[Example~10]{KNPS06}.
On the other hand,
Karasev~\cite{Ka00} proved that $\tau(\F) \le 3 \cdot \nu(\F) = 3$
for any family $\F$ of pairwise-intersecting translates of a convex body
in the plane.
Also, for any family $\F$ of congruent disks such that $\nu(\F) = 2$,
Kim \etal~\cite{KNPS06} confirmed that $\tau(\F) \le 3 \cdot \nu(\F) = 6$.
Our Corollary~\ref{cor:parallelepipedt} confirms that
$\tau(\F) \le 2 \cdot \nu(\F)$
for any family $\F$ of translates of a parallelogram.
The following theorem
confirms the upper bound $\tau(\F) \le 3 \cdot \nu(\F)$
for another special case:

\begin{theorem}\label{thm:hexagon}
For any family $\F$ of translates of
a centrally symmetric convex hexagon,
$\tau(\F) \le 3 \cdot \nu(\F)$.
Moreover, if $\nu(\F) = 1$, then $\tau(\F) \le 2$.
\end{theorem}

A hexagon $p_1p_2p_3p_4p_5p_6$ is affinely regular if and only if
(i) it is centrally symmetric and convex, and 
(ii) $\overrightarrow{p_2p_1} + \overrightarrow{p_2p_3} = \overrightarrow{p_3p_4}$.
Note that a centrally symmetric convex hexagon is not necessarily
affinely regular. 
Gr\"unbaum~\cite{Gr59} showed that $\alpha_1(C) = 2$ for any affinely
regular hexagon $C$. 
Theorem~\ref{thm:hexagon} implies a stronger and more general result that
$2 = \alpha_1(C) \le \alpha(C) \le 3$
for any centrally symmetric convex hexagon $C$.
Theorem~\ref{thm:greedyt} (i), (ii), and (iii) imply that
$\alpha(C) \le 2$ for any square $C$,
$\alpha(C) \le 5$ for any triangle $C$,
and
$\alpha(C) \le 4$ for any disk $C$.
Theorem~\ref{thm:greedyh} (i), (ii), and (iii) imply that
$\beta(C) \le 4$ for any square $C$,
$\beta(C) \le 12$ for any triangle $C$,
and
$\beta(C) \le 7$ for any disk $C$.
We also have the lower bounds
$\beta(C) \ge \alpha(C) \ge \frac32 $ for any square $C$~\cite{GL85},
$\beta(C) \ge \alpha(C) \ge \alpha_1(C) = 3$ for any triangle $C$~\cite{CS67},
$\alpha(C) \ge \alpha_1(C) = 3$ and
$\beta(C) \ge \beta_1(C) = 4$ for any disk $C$~\cite{Gr59,Da86}.
Table~\ref{tab2} summarizes
the current best bounds on $\alpha(C)$ and $\beta(C)$
for some special convex bodies $C$ in the plane.
 
\begin{table}[htbp]
\centering
\begin{tabular}{l|lr|lr|lr|lr}
\hline
\textsl{Special convex body $C$ in the plane}
	& \multicolumn{2}{c|}{\textsl{$\alpha(C)$ lower}}
	& \multicolumn{2}{c|}{\textsl{$\alpha(C)$ upper}}
	& \multicolumn{2}{c|}{\textsl{$\beta(C)$ lower}}
	& \multicolumn{2}{c}{\textsl{$\beta(C)$ upper}}
\\
\hline
centrally symmetric convex hexagon
	& $2$ & \cite{Gr59}
	& $3$ & T\ref{thm:hexagon}
	& $2$ & \cite{Gr59}
	& $7$ & T\ref{thm:greedyh}
\\
square
	& $\frac32$ & \cite{GL85}
	& $2$ & T\ref{thm:greedyt}
	& $\frac32$ & \cite{GL85}
	& $4$ & T\ref{thm:greedyh}
\\
triangle
	& $3$ & \cite{CS67}
	& $5$ & T\ref{thm:greedyt}
	& $3$ & \cite{CS67}
	& $12$ & T\ref{thm:greedyh}
\\
disk
	& $3$ & \cite{Gr59}
	& $4$ & T\ref{thm:greedyt}
	& $4$ & \cite{Gr59}
	& $7$ & T\ref{thm:greedyh}
\\
\hline
\end{tabular}
\caption{\small Lower and upper bounds on $\alpha(C)$ and $\beta(C)$
for some special convex bodies $C$ in the plane.} 
\label{tab2}
\end{table}

\section{Upper bound for translates of an arbitrary convex body in $\RR^d$}
\label{sec:conv}

In this section we prove Theorem~\ref{thm:conv}.
Let $\F$ be a family of translates of a convex body $C$ in $\RR^d$.
Let $P$ and $Q$ be any two parallelepipeds in $\RR^d$
that are parallel to each other, such that $P \subseteq C \subseteq Q$.
Since the two values $\tau(\F)$ and $\nu(\F)$ are invariant
under any non-singular affine transformation of $C$,
we can assume that $P$ and $Q$
are axis-parallel and have edge lengths $1$ and $e_i$, respectively,
along the axis $x_i$, $1\le i \le d$.

We first show that $\tau(\T) \le \lceil e_d \rceil \cdot \nu(\T)$
for any family $\T$ of $C$-translates whose corresponding $P$-translates
intersect a common line $\ell$ parallel to the axis $x_d$.
Define the \emph{$x_d$-coordinate} of a $C$-translate
as the smallest $x_d$-coordinate of a point in the corresponding $P$-translate.
Set $\T_1 = \T$,
let $C_1$ be the $C$-translate in $\T_1$ with the smallest $x_d$-coordinate,
and let $\S_1$ be the subfamily of $C$-translates in $\T_1$
that intersect $C_1$ ($\S_1$ includes $C_1$ itself).
Then, for increasing values of $i$,
while $\T_i = \T \setminus \bigcup_{j=1}^{i-1} \S_j$ is not empty,
let $C_i$ be the $C$-translate in $\T_i$ with the smallest $x_d$-coordinate,
and let $\S_i$ be the subfamily of $C$-translates in $\T_i$
that intersect $C_i$.
The iterative process ends with a partition $\T = \bigcup_{i=1}^m \S_i$,
where $m \le \nu(\T)$.

Denote by $c_i$ the $x_d$-coordinate of $C_i$.
Then each $C$-translate in the subfamily $\S_i$,
which is contained in a $Q$-translate of edge length $e_d$
along the axis $x_d$,
has an $x_d$-coordinate of at least $c_i$ and at most $c_i + e_d$,
and the corresponding $P$-translate,
whose edge length along the axis $x_d$ is $1$,
contains at least one of the $\lceil e_d \rceil$ points on $\ell$
with $x_d$-coordinates $c_i + 1, \ldots, c_i + \lceil e_d \rceil$.
These $\lceil e_d \rceil$ points form a piercing set for $\S_i$, 
hence $\tau(\S_i) \le \lceil e_d \rceil$. It follows that
\begin{equation}\label{eq:1}
\tau(\T) \le \sum_{i=1}^{m} \tau(\S_i)
	\le \lceil e_d \rceil \cdot m
	\le \lceil e_d \rceil \cdot \nu(\T).
\end{equation}

For $(a_1,\ldots,a_{d-1}) \in \RR^{d-1}$,
denote by $\ell(a_1,\ldots,a_{d-1})$ the following line in $\RR^d$
that is parallel to the axis $x_d$:
$$
\{\,
(x_1,\ldots,x_d) \mid
(x_1,\ldots,x_{d-1}) = (a_1,\ldots,a_{d-1})
\,\}.
$$
Now consider the following (infinite) set $\L$ of parallel lines:
$$
\{\,
\ell(j_1 + b_1, \ldots, j_{d-1} + b_{d-1})
	\mid (j_1,\ldots,j_{d-1}) \in \ZZ^{d-1}
\,\},
$$
where $(b_1,\ldots,b_{d-1}) \in \RR^{d-1}$
is chosen such that no line in $\L$ is tangent to the $P$-translate
of any $C$-translate in $\F$.
Recall that $P$ and $Q$ are axis-parallel
and have edge lengths $1$ and $e_i$, respectively,
along the axis $x_i$, $1\le i \le d$.
So we have the following two properties:
\begin{enumerate}
\item
For any $C$-translate in $\F$,
the corresponding $P$-translate intersects exactly one line in $\L$.
\item
For any two $C$-translates in $\F$,
if the two corresponding $P$-translates intersect
two different lines in $\L$ of distance at least $e_i + 1$
along some axis $x_i$, $1\le i \le d-1$,
then the two $C$-translates are disjoint.
\end{enumerate}

Partition $\F$ into subfamilies $\F(j_1,\ldots,j_{d-1})$
of $C$-translates whose corresponding $P$-translates
intersect a common line $\ell(j_1 + b_1, \ldots, j_{d-1} + b_{d-1})$.
Let $\F'(k_1,\ldots,k_{d-1})$ be the union of the families
$\F(j_1,\ldots,j_{d-1})$ such that
$j_i \bmod \lceil e_i + 1 \rceil = k_i$ for $1 \le i \le d-1$.
It follows from \eqref{eq:1} that
the transversal number of each subfamily $\F'(k_1,\ldots,k_{d-1})$
is at most $\lceil e_d \rceil$ times its packing number.
Therefore we have
\begin{align}\label{eq:d}
\tau(\F)
\le \sum_{(k_1,\ldots,k_{d-1})} \tau\left( \F'(k_1,\ldots,k_{d-1}) \right)
\nonumber
&\le \lceil e_d \rceil
	\sum_{(k_1,\ldots,k_{d-1})} \nu\left( \F'(k_1,\ldots,k_{d-1}) \right)
\nonumber
\\
&\le \left(
	\lceil e_d \rceil \prod_{i=1}^{d-1} \lceil e_i + 1 \rceil 
	\right) \cdot \nu(\F).
\end{align}

Since \eqref{eq:d} holds for any pair of parallelepipeds $P$ and $Q$ in $\RR^d$
that are parallel to each other and satisfy $P \subseteq C \subseteq Q$,
it follows by the definition of $\gamma(C)$ that
$\tau(\F) \le \gamma(C) \cdot \nu(\F)$.
By Lemma~\ref{lem:cs},
there indeed exist two such parallelepipeds $P$ and $Q$
with length ratios $\lambda_i(P,Q) = d$ for $1 \le i \le d$.
It then follows that $\gamma(C) \le d(d+1)^{d-1}$
for any convex body $C$ in $\RR^d$.
This completes the proof of Theorem~\ref{thm:conv}.

\section{Upper bound for translates of a centrally symmetric convex body in $\RR^d$}
\label{sec:symm}

In this section we prove Theorem~\ref{thm:symm}.
Recall that $|C|$ is the Lebesgue measure of a convex body $C$ in $\RR^d$,
and that $|\F|$ is the Lebesgue measure of the union of
a family $\F$ of convex bodies in $\RR^d$.
To establish the desired bound on $\tau(\F)$ in terms of $\nu(\F)$
for any family $\F$ of translates
of a centrally symmetric convex body $S$ in $\RR^d$,
we link both $\tau(\F)$ and $\nu(\F)$ to the ratio $|\F|/|S|$.
We first prove a lemma that links the transversal number $\tau(\F)$
to the ratio $|\F|/|S|$ via the lattice covering density of $S$:

\begin{lemma}\label{lem:tau}
Let $\F$ be a family of translates of a centrally symmetric convex body $S$
in $\RR^d$.
If there is a lattice covering of $\RR^d$ with translates of $S$
whose covering density is $\theta$, $\theta \ge 1$,
then $\tau(\F) \le \theta \cdot |\F|/|S|$.
\end{lemma}

\begin{proof}
Denote by $S_p$ a translate of the convex body $S$ centered at a point $p$.
Since $S$ is centrally symmetric,
for any two points $p$ and $q$,
$p$ intersects $S_q$ if and only if $q$ intersects $S_p$.
Given a lattice covering of $\RR^d$ with translates of $S$,
every point $p \in \RR^d$ is contained in some translate $S_q$
in the lattice covering,
hence every translate $S_p$ contains some lattice point $q$.

Let $\Lambda$ be a lattice such that
the corresponding lattice covering with translates of $S$
has a covering density of $\theta$.
Divide the union of the convex bodies in $\F$ into pieces
by the cells of the lattice $\Lambda$,
then translate all cells (and the pieces)
to a particular cell, say $\sigma$.
By the pigeonhole principle, there exists a point in $\sigma$, say $p$,
that is covered at most $\lfloor |\F|/|\sigma| \rfloor$ times
by the overlapping pieces of the union.
Let $k$ be the number of times that $p$ is covered by the pieces.
Now fix $\F$ but translate the lattice $\Lambda$ to $\Lambda'$
until $p$ becomes a lattice point of $\Lambda'$.
Then exactly $k$ lattice points of $\Lambda'$
are covered by the $S$-translates in $\F$.
Since every $S$-translate in $\F$ contains some lattice point of $\Lambda'$,
we have obtained a transversal of $\F$ consisting of
$k \le \lfloor |\F|/|\sigma| \rfloor$ lattice points of $\Lambda'$.
Note that $\theta = |S|/|\sigma|$,
and the proof is complete.
\end{proof}

The following lemma\footnote{The planar case of Lemma~\ref{lem:nu}
is also implied by~\cite[Theorem~5]{BDJ08b}.}
is a dual of the previous lemma,
and links the packing number $\nu(\F)$
to the ratio $|\F|/|S|$ via the lattice packing density of $S$:

\begin{lemma}\label{lem:nu}
Let $\F$ be a family of translates of a centrally symmetric convex body $S$
in $\RR^d$.
If there is a lattice packing in $\RR^d$ with translates of $S$
whose packing density is $\delta$, $\delta \le 1$,
then $\nu(\F) \ge \frac{\delta}{2^d} \cdot |\F|/|S|$.
\end{lemma}

\begin{proof}
Let $S'$ be a homothet of $S$ scaled up by a factor of $2$.
Since $S$ is centrally symmetric,
an $S$-translate is contained by an $S'$-translate if and only if
the $S$-translate contains the center of the $S'$-translate.
Given a lattice packing in $\RR^d$ with translates of $S'$,
two $S'$-translates centered at two different lattice points are disjoint,
hence two $S$-translates containing two different lattice points
are disjoint.

Let $\Lambda$ be a lattice such that
the corresponding lattice packing with translates of $S'$
has a packing density of $\delta$
(such a lattice exists because $S'$ is homothetic to $S$).
Divide the union of the convex bodies in $\F$ into pieces
by the cells of the lattice $\Lambda$,
then translate all cells (and the pieces)
to a particular cell, say $\sigma$.
By the pigeonhole principle, there exists a point in $\sigma$, say $p$,
that is covered at least $\lceil |\F|/|\sigma| \rceil$ times
by the overlapping pieces of the union.
Let $k$ be the number of times that $p$ is covered by the pieces.
Now fix $\F$ but translate the lattice $\Lambda$ to $\Lambda'$
until $p$ becomes a lattice point of $\Lambda'$.
Then exactly $k$ lattice points of $\Lambda'$
are covered by the $S$-translates in $\F$.
Choose $k$ translates in $\F$,
each containing a distinct lattice point of $\Lambda'$.
Since any two $S$-translates
containing two different lattice points of $\Lambda'$ are disjoint,
we have obtained a subset of $k \ge \lceil |\F|/|\sigma| \rceil$
pairwise-disjoint $S$-translates in $\F$.
Note that $\delta = |S'|/|\sigma| = 2^d |S|/|\sigma|$,
and the proof is complete.
\end{proof}

By Lemma~\ref{lem:tau} and Lemma~\ref{lem:nu} we have,
for any family $\F$ of translates of a centrally symmetric convex body
in $\RR^d$,
$$
\tau(\F) \le \theta_L(S) \cdot \frac{|\F|}{|S|}
	=  2^d \cdot \frac{\theta_L(S)}{\delta_L(S)} \cdot
		\frac{\delta_L(S)}{2^d} \cdot \frac{|\F|}{|S|}
	\le  2^d \cdot \frac{\theta_L(S)}{\delta_L(S)} \cdot
		\nu(\F).
$$
Smith~\cite{Sm06} proved that,
for any centrally symmetric convex body $S$ in 3-space,
$\theta_L(S) \le 3 \cdot \delta_L(S)$.
This immediately implies that,
for any family $\F$ of translates of a centrally symmetric convex body $S$
in 3-space,
$\tau(\F) \le 2^3 \cdot 3 \cdot \nu(\F) = 24 \cdot \nu(\F)$.
A similar inequality for the planar case was proved by Kuperberg~\cite{Ku87}:
for any (not necessarily centrally symmetric) convex body $C$ in the plane,
$\theta(C) \le \frac43 \cdot \delta(C)$.
However, this result is not about lattice covering and packing,
so we cannot use it to obtain the bound in Theorem~\ref{thm:symm}
for the planar case. Instead, we prove the following ``sandwich'' lemma:

\begin{lemma}
Let $\F$ be a family of translates of a (not necessarily centrally symmetric)
convex body $C$ in $\RR^d$.
Let $A$ and $B$ be two centrally symmetric convex bodies in $\RR^d$ such that
$A \subseteq C \subseteq B$. Then
$$
\tau(\F) \le 2^d \cdot \frac{|B|}{|A|} \cdot \frac{\theta_L(A)}{\delta_L(B)}
\cdot \nu(\F).
$$
\end{lemma}

\begin{proof}
Since $A \subseteq C$, it follows by Lemma~\ref{lem:tau} that
$$
\tau(\F) \le \theta_L(A) \cdot \frac{|\F|}{|A|}.
$$
Since $C \subseteq B$, it follows by Lemma~\ref{lem:nu} that
$$
\nu(\F) \ge \frac{\delta_L(B)}{2^d} \cdot \frac{|\F|}{|B|}.
$$
Putting these together yields
\[
\tau(\F)
\le \theta_L(A) \cdot \frac{|\F|}{|A|}
=   2^d \cdot \frac{|B|}{|A|} \cdot \frac{\theta_L(A)}{\delta_L(B)}
	\cdot \frac{\delta_L(B)}{2^d} \cdot \frac{|\F|}{|B|}
\le 2^d \cdot \frac{|B|}{|A|} \cdot \frac{\theta_L(A)}{\delta_L(B)}
	\cdot \nu(\F).
\qedhere
\]
\end{proof}

We also need the following lemma which is now
folklore~\cite[Theorem~2.5 and Theorem~2.8]{BMP05}:

\begin{lemma}
For any centrally symmetric convex body $S$ in the plane,
there are two centrally symmetric convex hexagons $H$ and $H'$ such that
$H \subseteq S \subseteq H'$ and $|H|/|H'| \ge 3/4$.
\end{lemma}

Note that $\theta_L(H) = \delta_L(H) = 1$
for a centrally symmetric convex hexagon $H$.
Set $A = H$, $B = H'$, and $C = S$ in the previous two lemmas,
and we have,
for any family $\F$ of translates of a centrally symmetric convex body
in the plane,
$$
\tau(\F) \le 2^2 \cdot \frac43 \cdot \frac11 \cdot \nu(\F)
	= \frac{16}3 \cdot \nu(\F).
$$
This completes the proof of Theorem~\ref{thm:symm}.

\section{Upper bound by greedy decomposition and lower bound by packing and covering}

In this section we prove Theorems \ref{thm:greedyt}, \ref{thm:greedyh},
and \ref{thm:lower}.

\paragraph{Proof of Theorem~\ref{thm:greedyt}.}

Let $\F$ be a family of translates of a convex body $C$ in $\RR^d$.
Without loss of generality, assume that $\kappa((C-C)\cap L, C)$ is minimized
when $L = \{ (x_1,\ldots,x_d) \mid x_d \ge 0 \}$.
Perform a \emph{greedy decomposition} as follows.
For $i = 1,2,\ldots$,
while $\T_i = \F \setminus \bigcup_{j=1}^{i-1} \S_j$ is not empty,
let $C_i$ be the translate of $C$ in $\T_i$ that contains a point
of the largest $x_d$-coordinate,
and let $\S_i$ be the subfamily of translates in $\T_i$
that intersect $C_i$ ($\S_i$ includes $C_i$ itself).
The iterative process ends with a partition $\F = \bigcup_{i=1}^m \S_i$,
where $m \le \nu(\F)$.
We next show that $\tau(\S_i) \le \kappa((C-C)\cap L,C)$.

Choose any point in $C$ as a reference point.
We have the following lemma:

\begin{lemma}\label{lem:ccc}
Let $A$ and $B$ be two translates of $C$ with reference points $a$
and $b$, respectively.
Then,
\begin{enumerate}
\item[\textup{(i)}]
$A$ contains $b$ if and only if $-(B-b)+b$ contains $a$,
\item[\textup{(ii)}]
If $A$ intersects $B$,
then $a$ is contained in a translate of $C-C$ centered at $b$.
\end{enumerate}
\end{lemma}

\begin{proof}
(i)
$b \in A \iff b-a \in A-a = B-b \iff a-b \in -(B-b) \iff a \in -(B-b)+b$.
(ii)
Let $c \in A \cap B$. Then
$c \in A \Longrightarrow c-a \in A-a \Longrightarrow a-c \in -(A-a)$,
and
$c \in B \Longrightarrow c-b \in B-b = A-a$.
It follows that
$a-b = (a-c) + (c-b) \in -(A-a) + (A-a) = C-C$.
\end{proof}

By Lemma~\ref{lem:ccc} (ii),
the reference point of each translate of $C$ in $\S_i$
is contained in a translate of $C-C$
centered at the reference point of $C_i$.
Since the translate of $C-C$ is covered by
$\kappa(C-C,-C)$ translates of $-C$,
it follows by Lemma~\ref{lem:ccc} (i) that
each translate of $C$ in $\S_i$ contains one of the $\kappa(C-C,-C)$
corresponding reference points.
Therefore,
\begin{equation}\label{eq:CCC}
\tau(\S_i) \le \kappa(C-C,-C) = \kappa(C-C,C).
\end{equation}
The stronger bound $\tau(\S_i) \le \kappa((C-C)\cap L,C)$ follows
by our choice of $C_i$.
We have
$$
\tau(\F) \le \sum_{i=1}^m \tau(\S_i) \le \kappa((C-C)\cap L,C) \cdot m
	\le \kappa((C-C)\cap L,C) \cdot \nu(\F).
$$

In the special case that $C$ is a centrally symmetric convex body in the plane,
$C-C$ is a translate of $2C$.
Assume without loss of generality that $C$ is centered at the origin.
Then $C-C = 2C$.
We have the following lemma on covering $2C$ with translates of $C$;
this lemma is implicit in a result by Gr\"unbaum~\cite[Theorem~4]{Gr59},
we nevertheless present our own simple proof here for completeness:

\begin{lemma}\label{lem:c7}
Let $C$ be a centrally symmetric convex body in the plane.
Then $2C$ can be covered by seven translates of $C$,
including one translate concentric with $2C$ and six others centered
at the six vertices, respectively,
of an affinely regular hexagon $H_C$ concentric with $2C$.
\end{lemma}

\begin{figure}[htbp]
\psfrag{o}{$o$}
\psfrag{o'}{$o'$}
\psfrag{p1}{$p_1$}
\psfrag{p2}{$p_2$}
\psfrag{p3}{$p_3$}
\psfrag{p4}{$p_4$}
\psfrag{p5}{$p_5$}
\psfrag{p6}{$p_6$}
\psfrag{q1}{$q_1$}
\psfrag{q6}{$q_6$}
\centering\includegraphics{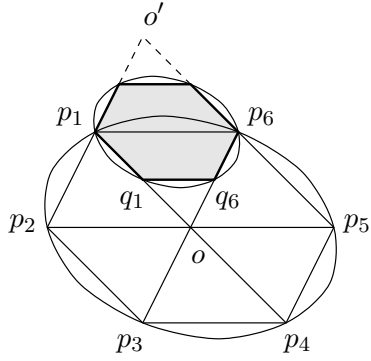}
\caption{\small
Covering $2C$ with seven translates of $C$.
$2H=p_1p_2p_3p_4p_5p_6$ is an affinely regular hexagon inscribed in $2C$;
$o$ is the center of $2C$;
$o'$ is the intersection of the two lines extending $p_2p_1$ and $p_5p_6$;
$q_1$ and $q_6$ are the midpoints of $p_1 o$ and $p_6 o$, respectively.}
\label{fig:c7}
\end{figure}

\begin{proof}
Refer to Figure~\ref{fig:c7}.
Let the center $o$ of $C$ be the origin.
Let $p_2$ and $p_5$ be the intersections of the boundary of $2C$
and an arbitrary line $\ell$ through the origin.
Choose two points $p_1$ and $p_6$ on the boundary of $2C$
on one side of the line $\ell$,
and choose two points $p_3$ and $p_4$ on the other side,
such that
$\overrightarrow{p_1p_6} = \overrightarrow{p_3p_4}
	= \frac12\, \overrightarrow{p_2p_5}$.
Then $p_1p_2p_3p_4p_5p_6$ is an affinely regular hexagon.
Let $2H$ be this hexagon inscribed in $2C$.
Consider the (shaded) hexagon $H'$ that is a translate of $H$
with two opposite vertices $p_1$ and $p_6$.
Let $q_1$ and $q_6$ be the midpoints of $p_1 o$ and $p_6 o$, respectively.
Then $q_1$ and $q_6$ are also vertices of $H'$.
The two hexagons $2H$ and $H'$ are homothetic with ratio 2
and with homothety center at the intersection $o'$
of the two lines extending $p_2p_1$ and $p_5p_6$.
Let $C'$ be a translate of $C$ such that $H'$ is inscribed in $C'$.
Then $C'$ covers the part of $2C$ between the two rays
$\overrightarrow{op_1}$
and
$\overrightarrow{op_6}$.
It follows that $2C$ is covered by seven translates of $C$,
one centered at the origin,
and six others centered at the midpoints of the six sides of $2H$, respectively.
The six midpoints are clearly the vertices of another (smaller)
affinely regular hexagon concentric with $2C$.
Let $H_C$ be this hexagon, and the proof is complete.
\end{proof}

Choose the halfplane $L$ through the center of $2C$ and any two opposite
vertices of the hexagon $2H=p_1p_2p_3p_4p_5p_6$ in Lemma~\ref{lem:c7}.
Then $\kappa((C-C)\cap L,C) \le 4$.
It follows that $\tau(\F) \le 4 \cdot \nu(\F)$
for any family $\F$ of translates of a centrally symmetric convex body
in the plane.

\begin{figure}[htbp]
\centering
\resizebox{0.67\linewidth}{!}{%
\parbox{3in}{\centering
\includegraphics{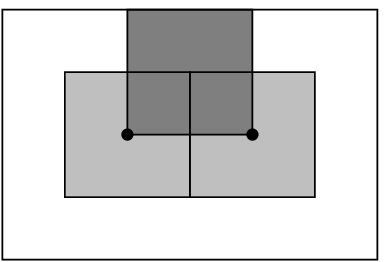}\\(a)\bigskip\\
\includegraphics{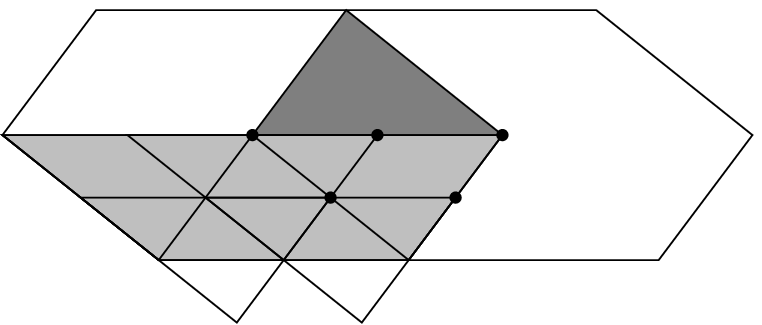}\\(b)}
\parbox{3in}{\centering
\includegraphics{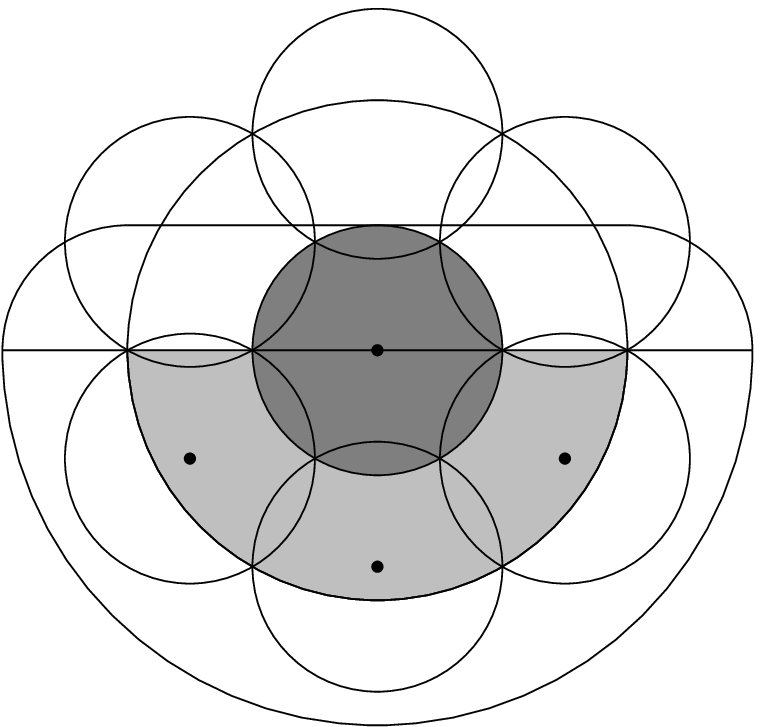}\\(c)}}
\caption{\small
Piercing a subfamily $\S_i$ of translates that intersect
the highest translate $C_i$ (dark-shaded).
(a)
The centers of the squares
are contained in the light-shaded rectangle;
the squares can be pierced by two points.
(b)
The lower-left vertices of the triangles
are contained in the light-shaded trapezoid;
the triangles can be pierced by five points.
(c)
The centers of the disks
are contained in the light-shaded half-disk;
the disks can be pierced by four points.}
\label{fig:plane}
\end{figure}

To complete the proof of Theorem~\ref{thm:greedyt},
we apply the greedy decomposition algorithm
to some simple types of convex bodies in the plane:
squares, triangles, and disks.
We use some known bounds on $\tau(\F)$
for families $\F$ with small $\nu(\F)$,
for example, $\alpha_1(C)$ for $\nu(\F) = 1$,
to obtain slightly better upper bounds for these special cases.
We refer to Figure~\ref{fig:plane},
where the $x_1$ and $x_2$ axes are the $x$ and $y$ axes.

First let $C$ be a square, and refer to Figure~\ref{fig:plane} (a).
Corollary~\ref{cor:parallelepipedt} implies that $\tau(\F) \le 2 \cdot \nu(\F)$
for any family $\F$ of translates of $C$.
We obtain a slightly better bound by a tighter analysis
of the greedy decomposition algorithm.
Assume that $C$ is axis-parallel and has side length $1$.
Choose the center of $C$ as the reference point.
Then the centers of the squares in $\S_i$
are contained in the light-shaded rectangle of width $2$ and height $1$,
which is covered by two unit squares centered at the two lower vertices of
$C_i$.
Each square in $\S_i$ contains one of the two lower vertices of $C_i$,
thus $\tau(\S_i) \le 2$.
Consider two cases:

\begin{enumerate}
\item $m \le \nu(\F) - 1$. Then
$$
\tau(\F) \le \sum_{i=1}^m \tau(\S_i)
	\le 2 \cdot(\nu(\F) - 1)
	= 2 \cdot\nu(\F) - 2.
$$
\item
$m = \nu(\F)$.
Then $\nu(\S_m) = 1$.
It follows that $\tau(\S_m) \le \alpha_1(C) = 1$~\cite{Gr59}.  Then
$$
\tau(\F) \le \sum_{i=1}^m \tau(\S_i)
	\le 2 \cdot(\nu(\F) - 1) + 1
	= 2 \cdot\nu(\F) - 1.
$$
\end{enumerate}

Next let $C$ be a triangle, and refer to Figure~\ref{fig:plane} (b).
Assume that $C$ has a horizontal lower side.
Choose the lower-left vertex of $C$ as the reference point.
The lower-left vertices of the triangles in $\S_i$
are contained in the light-shaded trapezoid,
which can be covered by five translates of $-C$.
Hence each triangle in $\S_i$ contains one of the upper-right vertices
of these five translates, thus $\tau(\S_i) \le 5$.
The proof can be finished in the same way as for squares
by considering the two cases $m \le \nu(\F) - 1$ and $m = \nu(\F)$,
and using the fact that $\alpha_1(C) = 3$ for any triangle $C$~\cite{CS67}.

Finally let $C$ be a disk, and refer to Figure~\ref{fig:plane} (c).
Assume that $C$ has radius $1$. Choose the center of $C$ as the reference point.
Then the centers of the disks in $\S_i$
are contained in the light-shaded half-disk of radius $2$.
It is well known (see~\cite{Fr}) that
a disk of radius $2$ can be covered by seven disks of radius $1$,
with one disk in the middle
and six others around in a hexagonal formation.
Therefore the half-disk of radius $2$
can be covered by four disks of radius $1$.
The center of each disk in $\S_i$ is contained by one of the four disks;
by symmetry, each disk in $\S_i$ contains the center of one of the four disks,
thus $\tau(\S_i) \le 4$.
Again, the proof can be finished by considering the two cases $m \le
\nu(\F) - 1$ and $m = \nu(\F)$ as done for squares and triangles,
and using the fact that $\alpha_1(C) = 3$ for any disk $C$~\cite{Gr59}.
Indeed the same argument shows that
$\tau(\F) \le 4\cdot (\nu(\F) - 1) + 3 = 4\cdot \nu(\F) - 1$
for any centrally symmetric convex body $C$ in the plane
since $\alpha_1(C) \le 3$ also holds~\cite{Ka00}.
This completes the proof of Theorem~\ref{thm:greedyt}.

\paragraph{Proof of Theorem~\ref{thm:greedyh}.}

Let $\F$ be a family of homothets of a convex body $C$ in $\RR^d$. 
We again use greedy decomposition.
The only difference in the algorithm is that $C_i$ is now chosen as
the smallest homothet of $C$ in $\T_i$.
By our choice of $C_i$,
each homothet in $\S_i$ contains a translate of $C_i$ that intersects $C_i$.
Hence the bound $\tau(\S_i) \le \kappa(C-C,C)$ follows in a similar way
as the derivation of \eqref{eq:CCC}.

Let now $C$ be a centrally symmetric convex body in the plane.
By Lemma~\ref{lem:c7},
we have $\kappa(C-C,C) \le 7$.
Then $\tau(\S_i) \le \kappa(C-C,C) \leq 7$,
from which it follows that  $\tau(\F) \le 7 \cdot \nu(\F)$
for any centrally symmetric convex body $C$ in the plane. 

The analysis for special types of convex bodies $C$ in the plane
(squares, triangles, and disks)
is also similar to the corresponding analysis
in the proof of Theorem~\ref{thm:greedyt}.
We obtain the bound $\tau(\S_i) \le \kappa(C-C,C)$
and show that
$\kappa(C-C,C) \le 4$ for any square $C$,
$\kappa(C-C,C) \le 12$ for any triangle $C$,
and $\kappa(C-C,C) \le 7$ for any disk $C$,
then use $\beta_1(C)$ instead of $\alpha_1(C)$ to bound $\tau(\S_m)$ in case~2.
As discussed in the introduction, it is known that
$\beta_1(C) = 1$ for any square $C$~\cite{Gr59},
$\beta_1(C) = 3$ for any triangle $C$~\cite{CS67},
and $\beta_1(C) = 4$ for any disk $C$~\cite{Gr59,Da86}.
This completes the proof of Theorem~\ref{thm:greedyh}.

\paragraph{Proof of Theorem~\ref{thm:lower}.}

Let $C$ be a convex body in $\RR^d$ and $n$ be a positive integer.
We will show that $\beta(C) \ge \alpha(C) \ge \theta_T(C) / \delta_T(C)$
by constructing a family $\F_n$ of $n^{2d}$ translates of $C$, such that
\begin{equation}\label{eq:limit}
\lim_{n\to \infty} \frac{\tau(\F_n)}{\nu(\F_n)} \ge 
\frac{\theta_T(C)}{\delta_T(C)}.
\end{equation}

By Lemma~\ref{lem:cs},
there exist two homothetic parallelepipeds $P$ and $Q$ with ratio $d$
such that $P \subseteq C \subseteq Q$.
Without loss of generality (via an affine transformation),
we can assume that $P$ and $Q$ are axis-parallel hypercubes of side lengths
$1$ and $d$, respectively,
and that $P$ is centered at the origin.
Now choose the origin as the reference point of $C$.
Let $\F_n = \{ C + t \mid t \in T_n \}$
be a family of translates of $C$
corresponding to a set of $n^{2d}$ regularly placed reference points
$$
T_n = \{ (t_1/n, \ldots, t_d/n) \mid
	(t_1, \ldots, t_d) \in \ZZ^d, 1 \le t_1,\ldots,t_d \le n^2 \}.
$$
Denote by $H(\ell)$ any axis-parallel hypercube of side length $\ell$.

We first obtain an upper bound on $\nu(\F_n)$.
For each $C + t \in \F_n$,
we have $C + t \subseteq C + T_n \subseteq Q + T_n$.
Note that $Q + T_n$ is an axis-parallel hypercube
of side length exactly $n - \frac1n + d$.
Denote by $\delta_T(X, Y)$ the supremum of the packing density
of a domain $Y \subseteq \RR^d$ by translates of $X$.
By a volume argument, we have
\begin{equation}\label{eq:nuFn}
\nu(\F_n)
	\le \frac{\delta_T(C, Q + T_n) \cdot |Q + T_n|}{|C|}
	= \frac{\delta_T\big( C, H(n - \frac1n + d) \big)
		\cdot (n - \frac1n + d)^d}{|C|}.
\end{equation}

We next obtain a lower bound on $\tau(\F_n)$.
By Lemma~\ref{lem:ccc} (i),
piercing the family  $\F_n$ of translates of $C$ is equivalent to
covering the corresponding set $T_n$ of reference points by translates of $-C$.
Let $S_n$ be any set of points such that $T_n \subseteq -C + S_n$,
that is, $T_n$ is covered by
the set $\{ -C + s \mid s \in S_n \}$ of translates of $-C$.
We also have $-\frac1n P \subseteq -\frac1n C$ since $P \subseteq C$.
It follows that
$$
\textstyle
-\frac1n P + T_n \subseteq -\frac1n C + (-C + S_n) = {-(1 + \frac1n)C} + S_n.
$$
Thus $\tau(\F_n)$ is at least the minimum number of translates of
${-(1 + \frac1n)C}$
that cover
$-\frac1n P + T_n$.
Note that
$-\frac1n P + T_n$
is an axis-parallel hypercube of side length exactly $n$.
Denote by $\theta_T(X, Y)$ the infimum of the covering density
of a domain $Y \subseteq \RR^d$ by translates of $X$.
Again by a volume argument, we have
\begin{equation}\label{eq:tauFn}
\tau(\F_n)
	\ge \frac{\theta_T\big( {-(1 + \frac1n)C}, {-\frac1n P} +  T_n \big)
		\cdot |{-\frac1n P} + T_n|}{|{-(1 + \frac1n)C}|}
	= \frac{\theta_T\big( (1 + \frac1n)C, H(n) \big)
		\cdot n^d}{(1 + \frac1n)^d \cdot |C|}.
\end{equation}

From the two inequalities \eqref{eq:nuFn} and \eqref{eq:tauFn},
it follows that
$$
\frac{\tau(\F_n)}{\nu(\F_n)}
	\ge \frac{\theta_T\big( (1 + \frac1n)C, H(n) \big)}
			{\delta_T\big( C, H(n - \frac1n + d) \big)}
		\cdot \frac{1}{(1 + \frac1n)^d (1 - \frac1{n^2} + \frac{d}n)^d}.
$$
Taking the limit as $n \to \infty$, we have
$\theta_T\big( (1 + \frac1n)C, H(n) \big) \to \theta_T(C)$,
$\delta_T\big( C, H(n - \frac1n + d) \big) \to \delta_T(C)$,
and $(1 + \frac1n)^d (1 - \frac1{n^2} + \frac{d}n)^d \to 1$.
This yields \eqref{eq:limit} as desired.

We now consider the special case that
$C$ is the $d$-dimensional unit ball $B^d$ in $\RR^d$.
We clearly have $\theta_T(B^d) \ge 1$.
Kabatjanski\u{\i} and Leven\v{s}te\u{\i}n~\cite{KL78}
showed that $\delta_T(B^d) = \delta(B^d) \le 2^{-(0.599 \pm o(1))d}$
as $d \to \infty$;
see also~\cite[p.~50]{BMP05}.
Therefore we have
$$
\beta(B^d) \ge \alpha(B^d) \ge \frac{\theta_T(B^d)}{\delta_T(B^d)}
	\ge 2^{(0.599 \pm o(1))d}
\textrm{ as } d\to\infty.
$$
This completes the proof of Theorem~\ref{thm:lower}.

\section{Upper bound for translates of a centrally symmetric convex hexagon}

In this section we prove Theorem~\ref{thm:hexagon}.
Let $\F$ be a family of translates of a centrally symmetric convex hexagon $H$
in the plane.

\begin{figure}[htbp]
\centering
\psfrag{p1}{$p_1$}
\psfrag{p2}{$p_2$}
\psfrag{p3}{$p_3$}
\psfrag{p4}{$p_4$}
\psfrag{p5}{$p_5$}
\psfrag{p6}{$p_6$}
\psfrag{p1'}{$p_1'$}
\psfrag{p4'}{$p_4'$}
\psfrag{u}{$u$}
\psfrag{v}{$v$}
\psfrag{u'}{$u'$}
\psfrag{v'}{$v'$}
\includegraphics{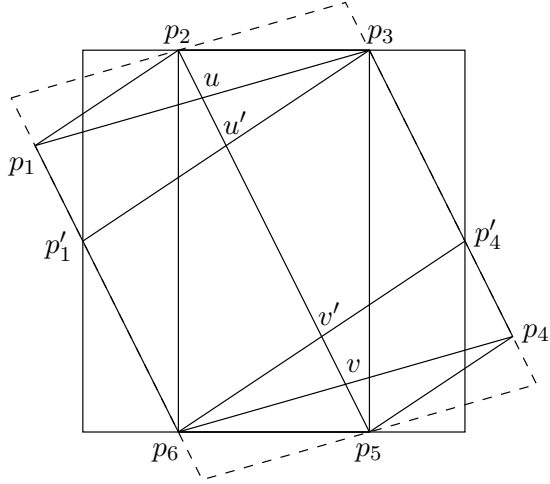}
\caption{\small A centrally symmetric convex hexagon $H = p_1p_2p_3p_4p_5p_6$.}
\label{fig:h3}
\end{figure}

We refer to Figure~\ref{fig:h3} for the general case $\nu(\F) \ge 1$.
We will prove that $\tau(\F) \le 3\cdot \nu(\F)$.
By Theorem~\ref{thm:conv},
it suffices to show that $\gamma(H) \le 3$.
We will show that $\gamma(H) \le 3$ by finding two parallelograms
$P$ and $Q$ that are parallel to each other, with length ratios
$w = \lambda_1(P,Q) \le 2$ and $h = \lambda_2(P,Q) = 1$,
such that $P \subseteq H \subseteq Q$.
Let $H = p_1p_2p_3p_4p_5p_6$.
Without loss of generality (via an affine transformation),
the parallelogram $p_2p_3p_5p_6$
is an axis-parallel rectangle of width $1/2$ and height $1$.
If the hexagon is contained in an axis-parallel unit square,
then we can choose $P$ and $Q$ as the rectangle $p_2p_3p_5p_6$ and the square,
whose length ratios are $w = 2$ and $h = 1$.
Suppose otherwise.
Assume that $p_1$ is higher than $p_4$.
Then we choose $P$ as the parallelogram $p_1p_3p_4p_6$ and $Q$ as the
(dashed) parallelogram circumscribing $H$ and parallel to $P$.
Let $u$ be the intersection of $p_1p_3$ and $p_2p_5$,
and let $v$ be the intersection of $p_4p_6$ and $p_2p_5$.
The length ratios of $P$ and $Q$ are $w = |p_2p_5|/|uv|$ and $h = 1$,
where $w$ is maximized to $2$ when $p_1$ and $p_4$ are the midpoints
of the two vertical sides of the unit square.

\begin{figure}[htbp]
\centering
\psfrag{S1}{$S_1$}
\psfrag{S2}{$S_2$}
\psfrag{S3}{$S_3$}
\psfrag{S1'}{$S_1'$}
\psfrag{S2'}{$S_2'$}
\psfrag{S3'}{$S_3'$}
\psfrag{1}{$1$}
\psfrag{2}{$2$}
\psfrag{3}{$3$}
\psfrag{H12}{$H_{12}$}
\psfrag{H13}{$H_{13}$}
\psfrag{H23}{$H_{23}$}
\psfrag{H}{$H$}
\psfrag{H'}{$H'$}
\resizebox{\linewidth}{!}{\includegraphics{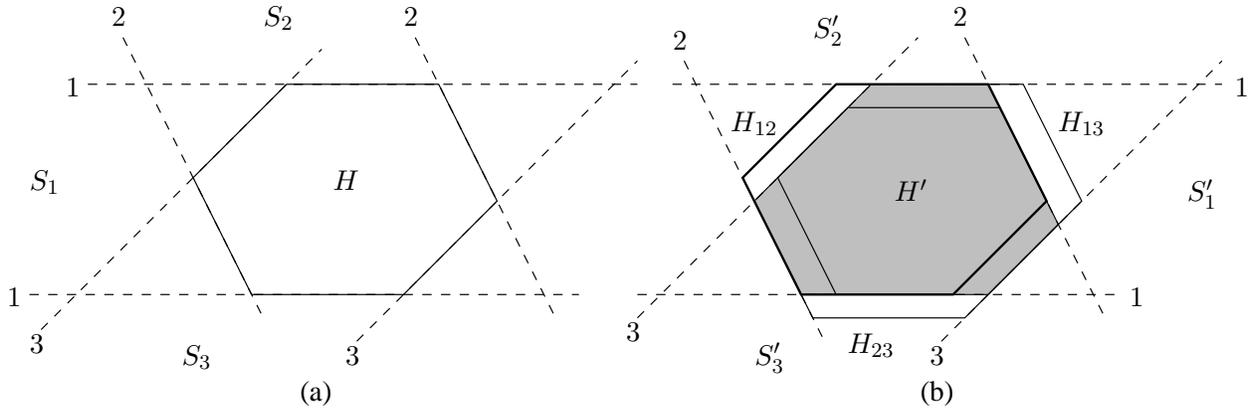}}
\\\hspace{\stretch1}(a)\hspace{\stretch2}(b)\hspace{\stretch1}
\caption{\small (a) A centrally symmetric convex hexagon $H$
is the intersection of three strips $S_1$, $S_2$, and $S_3$.
(b) The centers of all translates of $H$ in $\F$
are contained in the shaded hexagon $H'$ that is the intersection of
three strips $S_1'$, $S_2'$, and $S_3'$;
the shaded hexagon $H'$ is covered by any two of the three translates of $H$:
$H_{12} \subseteq S_1' \cap S_2'$,
$H_{13} \subseteq S_1' \cap S_3'$,
and $H_{23} \subseteq S_2' \cap S_3'$.
$H_{12}$ is shown in bold lines.}
\label{fig:h2}
\end{figure}

We refer to Figure~\ref{fig:h2} for the special case $\nu(\F) = 1$.
We will prove that $\tau(\F) \le 2$.
The centrally symmetric convex hexagon $H$ is the intersection of three strips
$S_1$, $S_2$, and $S_3$,
each bounded by the two supporting lines of a pair of parallel edges of $H$.
Without loss of generality, assume that the strip $S_1$ is horizontal.
Let $A$ be the highest translate of $H$ in $\F$,
and let $B$ be any other translate of $H$ in $\F$.
Then the $y$-coordinates of the centers of $A$ and $B$
differ by at most the width of the strip $S_1$.
This implies that
the centers of all translates of $H$ in $\F$ are contained
in a translate of $S_1$.
Apply the same argument to the other two strips $S_2$ and $S_3$.
It follows that the centers of all translates of $H$ in $\F$ are contained
a hexagon $H'$ that is the intersection of three strips
$S_1'$, $S_2'$, and $S_3'$, which are translates of $S_1$, $S_2$, and $S_3$,
respectively.
Let $H_{12}$, $H_{13}$, and $H_{23}$ be the three unique translates of $H$
contained in $S_1' \cap S_2'$, $S_1' \cap S_3'$, and $S_2' \cap S_3'$,
respectively.
Then any two of the three translates of $H$, say $H_{12}$ and $H_{13}$,
cover the hexagon $H'$.
It follows by symmetry that two points
(the centers of $H_{12}$ and $H_{13}$)
are enough to pierce all members of $\F$.
This completes the proof of Theorem~\ref{thm:hexagon}.

\section{Conclusion}

We believe that our bounds in Lemma~\ref{lem:tau} and Lemma~\ref{lem:nu}
are not tight. We have the following conjectures:

\begin{conjecture}\label{conjtau}
Let $\F$ be a family of translates of a centrally symmetric convex body $S$
in the plane. Then $\tau(\F) \le |\F|/|S|$.
\end{conjecture}

\begin{conjecture}\label{conjnu}
Let $\F$ be a family of translates of a centrally symmetric convex body $S$
in the plane. Then $\nu(\F) \ge \frac14 \cdot |\F|/|S|$.
\end{conjecture}

If both conjectures were to hold
(note that they hold for the special cases when $S$ is
a parallelogram or a centrally symmetric convex hexagon
since $\theta_L(S) = \delta_L(S) = 1$ in such cases),
then we would have an alternative proof of essentially the same bound
$\tau(\F) \le 4 \cdot \nu(\F)$
as in Theorem~\ref{thm:greedyt}
for any family $\F$ of translates of a centrally symmetric convex body
in the plane.
Conjecture~\ref{conjnu} is related to another recent
conjecture~\cite{BDJ08} in the spirit of Rado~\cite{Ra49}:

\begin{conjecture}[Bereg, Dumitrescu, and Jiang~\cite{BDJ08}]\label{conj:disks}
For any set $\S$ of (not necessary congruent) closed disks in the plane,
there exists a subset $\I$ of pairwise-disjoint disks such that
$|\I|/|\F| \ge \frac14$.
\end{conjecture}

Note that a disk $D$ is centrally symmetric;
for any family $\F$ of congruent disks
(i.e., translates of a disk) in the plane,
$\nu(\F) \ge \frac14 \cdot |\F|/|D|$
if and only if there exists a subset $\I$ of pairwise-disjoint disks such that
$|\I|/|\F| \ge \frac14$.

\paragraph{Approximation algorithms.}
A computational problem related to the results of this paper is
finding a minimum-cardinality point set
that pierces a given set of geometric objects.
This problem is NP-hard even for the special case of
axis-parallel unit squares in the plane~\cite{FPT81},
and it admits a polynomial-time approximation scheme for the general case of fat objects in $\RR^d$~\cite{Ch03}
(see also~\cite{CKL08} for similar approximation schemes
for several related problems).
These approximation schemes
have very high time complexities $n^{O(1/\epsilon^d)}$,
and hence are impractical.
Our methods for obtaining the upper bounds in Theorems \ref{thm:conv},
\ref{thm:symm}, \ref{thm:greedyt}, and \ref{thm:greedyh} are constructive
and lead to efficient constant-factor approximation algorithms
for piercing a set of translates or homothets of a convex body.
The approximation factors,
which depend on the dimension $d$,
are the multiplicative factors in the respective bounds on $\tau(\F)$
in terms of $\nu(\F)$ in the theorems,
see also Table~\ref{tab1} and Table~\ref{tab2}.
For instance, Theorem~\ref{thm:conv} yields a factor-$6$ approximation algorithm
for piercing translates of a convex body in the plane, 
and Theorem~\ref{thm:greedyh} yields a factor-$216$ approximation algorithm
for piercing homothets of a convex body in $3$-space.

\paragraph{Note.}
After completion of this work and shortly before journal submission,
we learned that very recently, Nasz\'odi and Taschuk~\cite{NT09} independently 
obtained some results similar in nature to our Theorems~\ref{thm:greedyh}
and~\ref{thm:lower}.
There are however differences in the specific bounds:

\begin{enumerate}
\item
They proved\footnote{%
Nasz\'odi and Taschuk~\cite{NT09} used the terms
$d\log d + \log\log d + 5d$
instead of
$d\ln d + d\ln\ln d + 5d$
throughout their paper,
which are clearly misprints.
Recall that $\theta_T(C) < d\ln d + d\ln\ln d + 5d$
for any convex body $C$ in $\RR^d$~\cite{Ro57}.}
that
$\beta(C) \le 2^d{2d \choose d}(d\ln d + d\ln\ln d + 5d)$
for any convex body $C$ in $\RR^d$,
and that $\beta(C) \le 3^d(d\ln d + d\ln\ln d + 5d)$
for any centrally symmetric convex body $C$ in $\RR^d$.
Note that their upper bound for the centrally symmetric case
is essentially the same as our bound $\beta(C) \le 3^d \theta_T(C)$
by Theorem~\ref{thm:greedyh} and Lemma~\ref{lem:kappa}.
Their upper bound for the general case, however, is weaker
than our bound $\beta(C) \le \frac{2^d}{d+1} 3^{d+1} \theta_T(C)$,
also by Theorem~\ref{thm:greedyh} and Lemma~\ref{lem:kappa}.
By Stirling's formula,
${2d \choose d} = \frac{(2d)!}{(d!)^2} = \Theta(4^d/\sqrt{d})$.
Compare the factor
$2^d{2d \choose d} = \Theta(8^d/\sqrt{d})$ in their bound
with the factor
$\frac{2^d}{d+1} 3^{d+1} = \Theta(6^d/d)$ in our bound.

\item
They also derived the following lower bound:
for sufficiently large $d$, there
is a convex body $C$ in $\RR^d$ such that $\alpha(C) \ge \frac12(1.058)^d$.
This lower bound is analogous to our exponential lower bound
in Theorem~\ref{thm:lower}: if $C$ is the unit ball $B^d$ in $\RR^d$,
then $\alpha(C) \ge 2^{(0.599 \pm o(1))d} \approx (1.51)^d$ as $d\to\infty$.
Recall that our lower bound for the unit ball $B^d$ follows from 
a general lower bound for any convex body $C$ in $\RR^d$,
namely, $\alpha(C) \ge \frac{\theta_T(C)}{\delta_T(C)}$.
A comparison shows that their lower bound is both weaker
and less general than ours.
\end{enumerate}


\appendix

\section{Lower bound for translates of a triangle}\label{sec:t1}

In this section we prove the lower bound $\alpha_1(T) \ge 3$
for any triangle $T$ by a very simple\footnote{%
Simpler than the previous constructions~\cite{CS67,KT08}
that give the same lower bound.}
construction:

\begin{proposition}\label{prp:triangle}
For any triangle $T$, there exists a family $\F$ of nine translates of $T$
such that $\nu(\F) = 1$ and $\tau(\F) = 3$.
\end{proposition}

\begin{figure}[htbp]
\psfrag{A}{$A$}
\psfrag{B}{$B$}
\psfrag{C}{$C$}
\psfrag{a}{$a$}
\psfrag{b}{$b$}
\psfrag{c}{$c$}
\centering
\includegraphics{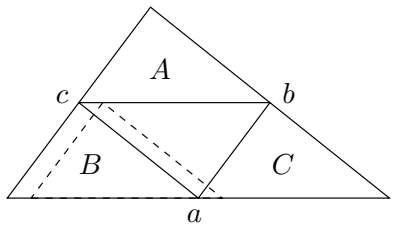}
\caption{\small Three pairwise-tangent translates $A$, $B$, and $C$ of
a triangle $T$. The dashed triangle is $B_C$.}
\label{fig:t3}
\end{figure}

\begin{proof}
We refer to Figure~\ref{fig:t3}.
Let $A$, $B$, and $C$ be three translates of $T$ that are pairwise-tangent
with intersections at three vertices $a$, $b$, and $c$.
We obtain six more translates of $T$ as follows.
Translate a copy of $T$ for a short distance $\epsilon$ from $B$ toward $C$,
and let $B_C$ be the resulting translate.
Similarly obtain $A_B$, $A_C$, $B_A$, $C_A$, and $C_B$.
Let $\F$ be the family of nine translates
$A$, $B$, $C$, $A_B$, $A_C$, $B_C$, $B_A$, $C_A$, and $C_B$.
It is clear that any two members of $\F$ intersect.
We next show that three points are necessary to pierce all members of $\F$.
Suppose for contradiction that
two points are enough.
Then one of the two points must be $a$, $b$, or $c$
since $A$, $B$, and $C$ are pairwise-tangent.
Assume that $a$ is one of the two points.
Then the other point must intersect the three translates $A$, $B_A$, and $C_A$
that do not contain the point $a$.
But these three translates do not have a common point
when $\epsilon$ is sufficiently small.
We have reached a contradiction.
\end{proof}

By repeating the configuration of nine translates
in Proposition~\ref{prp:triangle},
we can obtain a family $\F$ of $9\,\nu(\F)$ translates of a triangle
such that $\tau(\F) = 3\cdot \nu(\F)$ for any $\nu(\F) \ge 1$.

\end{document}